\documentclass[11pt]{article}
\usepackage{amsmath,amssymb,url,color}
\newtheorem{theorem}{Theorem}
\newtheorem{definition}{Definition}
\newtheorem{corollary}{Corollary}
\newtheorem{proposition}{Proposition}
\newtheorem{lemma}{Lemma}
\newtheorem{conjecture}{Conjecture}
\newtheorem{ex}{Example}
\newtheorem{remark}{Remark}

\newcommand{\FF}{{\mathbb F}}
\newcommand{\F}{{\mathbb F}_{2^n}}
\newcommand{\FB}{{\mathbb F}_2}

\renewcommand{\S}{Section}
\newenvironment{proof}{\begin{trivlist}\item[]{\em Proof. }}%
{\samepage\hfill$\diamond$\end{trivlist}}
\usepackage{slashbox}

\title{Differential properties of functions $\boldsymbol{x\mapsto x^{2^t-1}}$
--~extended version\footnote{ of the paper which will appear in 
{\em IEEE Transactions on Information Theory}} --}
\author{C\'eline Blondeau, Anne Canteaut and Pascale Charpin
\footnote{SECRET project-team - INRIA Paris-Rocquencourt,
Domaine de Voluceau, B.P.~105,   
78153 Le Chesnay Cedex, France.\hfill
Email: {\tt 
celine.blondeau@inria.fr}, {\tt anne.canteaut@inria.fr}, {\tt pascale.charpin@inria.fr}}}
\begin{document}
\maketitle

\begin{abstract}
We provide an extensive study of the differential properties of  
the functions $x\mapsto x^{2^t-1}$ over $\F$, for $1 <t<n$. We notably show that the differential spectra of these functions are determined by the number of roots of the linear polynomials $x^{2^t}+bx^2+(b+1)x$ where $b$ varies in $\F$.We prove a strong relationship between the differential spectra of $x\mapsto x^{2^t-1}$ and $x\mapsto x^{2^{s}-1}$ for $s= n-t+1$. As a direct consequence, this result enlightens a connection between the differential properties of the cube function and of the inverse function. We also determine the complete differential spectra of $x \mapsto x^7$ by means of the value of some
 Kloosterman sums, and of $x \mapsto x^{2^t-1}$ for 
$t \in \{\lfloor n/2\rfloor, \lceil n/2\rceil+1, n-2\}$.
\end{abstract}


\section{Introduction}
Differential cryptanalysis is the first statistical attack proposed for
 breaking iterated block ciphers. Its publication~\cite{Biham_Shamir91}
 then gave rise to numerous works which investigate the security offered
 by different types of functions regarding differential attacks. 
This security is quantified by the so-called {\em differential uniformity}
 of the Substitution box used in the cipher~\cite{Nyberg_Knudsen92}.
 Most notably, finding
 appropriate S-boxes which guarantee that the cipher using them resist 
differential attacks is a major topic for the last twenty years, see e.g.~\cite{CCZ,Edel_etal06,BCP06,Bracken_etal08,Browning_Dillon_McQuistan_Wolfe10}. 

Power functions, {\em i.e.}, monomial functions, form a class of
suitable candidates since they usually have a lower implementation
cost in hardwa\-re. Also, their particular algebraic structure makes the determination of their differential properties easier. However, there are only a few power functions for which we can prove that they have a low differential uniformity. Up to equivalence, there are two large families of such functions: a subclass of the quadratic power functions (a.k.a. Gold functions) and a subclass of the so-called Kasami functions. Both of these families contain some permutations which are APN over \(\F\) for odd~\(n\) and differentially \(4\)-uniform for even~\(n\). The other known power functions with a low differential uniformity correspond to ``sporadic'' cases in the sense that the corresponding exponents vary with~\(n\)~\cite{Hernando_McGuire11} and they do not belong to a large class: they correspond to the exponents defined by Welch~\cite{Dobbertin98a,Canteaut_Charpin_Dobbertin00a}, by Niho~\cite{Dobbertin98b,Xiang_Hollmann01}, by Dobbertin~\cite{Dobbertin00}, by Bracken and Leander~\cite{BL}, and to the inverse function~\cite{Nyberg93}. It is worth noticing that some of these functions seem to have different structures because they do not share the same differential spectrum. For instance, for a quadratic power function or a Kasami function, the differential spectrum has only two values, {\em i.e.}, the number of occurrences of each differential belongs to \(\{0, \delta\}\) for some \(\delta\)~\cite{BCC}. The inverse function has a very different behavior since its differential spectrum has three values, namely \(0\), \(2\) and \(4\) and, for each input difference, there is exactly one  differential which is satisfied four times. 

However, when classifying all functions with a low differential uniformity, it can be noticed that the family of all power functions \(x \mapsto x^{2^t-1}\) over \(\F\), with \(1< t <n\), contains several functions with a low differential uniformity. Most notably, it includes the cube function and the inverse function, and also \(x \mapsto x^{2^{(n+1)/2}-1}\) for \(n\) odd, which is the inverse of a quadratic function. At a first glance, this family of exponents may be of very small relevance because the involved functions have distinct differential spectra. Then, they are expected to have distinct structures. For this reason, one of the motivations of our study was to determine whether some link could be established between the differential properties of the cube function and of the inverse function. Our work then answers positively to this question since it exhibits a general relationship between the differential spectra of \(x \mapsto x^{2^t-1}\) and \(x \mapsto x^{2^{n-t+1}-1}\) over \(\F\). We also determine the complete differential spectra of some other exponents in this family.

The rest of the paper is organized as follows. Section~\ref{sec-boo} recalls some definitions and some general properties of the differential spectrum of monomial functions. Section~\ref{sec-ds} then focuses on the differential spectra of the monomials \(x \mapsto x^{2^t-1}\). First, the differential spectrum of any such function is shown to be determined by the number of roots of a family of linear polynomials.  Then, we exhibit a symmetry property for the exponents in this family: it is proved that the differential spectra of \(x \mapsto x^{2^t-1}\) and \(x \mapsto x^{2^{n-t+1}-1}\) over \(\F\) are closely related. In Section~\ref{sec-spec},
we  determine the whole differential spectrum of \(x \mapsto x^7\) over~\(\F\).
It is expressed by means of some Kloosterman sums, and explicitly computed using the work 
of Carlitz~\cite{Carlitz69}. We then derive the differential spectra of
 \(x \mapsto x^{2^{n-2}-1}\). Further, we study the functions
 \(x \mapsto x^{2^{\lfloor n /2 \rfloor}-1}\) and 
\(x \mapsto x^{2^{\lceil n /2 \rceil +1}-1}\). 
We finally end up with some conclusions.

\section{Preliminaries}
\label{sec-boo}
\subsection{Functions over $\boldsymbol{\F}$ and their derivatives}

 Any function
$F$ from \(\F\) into \(\F\) can be expressed as a univariate
polynomial in \(\F[X]\). The {\em univariate degree} of
the polynomial $F$ is, as usual, the maximal integer value
of its exponents. The {\em algebraic degree} of $F$ is the maximal
 Hamming weight of its exponents:
\[
{\rm deg}~\left(\sum_{i=0}^{2^n-1}\lambda_iX^i\right)=
\max~\{wt(i)~|~\lambda_i\neq 0~\},
\]
where $\lambda_i\in\F$ and the {\em Hamming weight} is calculated as 
follows :
\[
i=\sum_{j=0}^{n-1}i_j 2^j~\mbox{with}~i_j\in\{0,1\},~
wt(i)=\sum_{j=0}^{n-1}i_j.
\]
In this paper, we will identify a polynomial of $\F[X]$ with the
corresponding function over $\F$. For instance, $F \in \F[X]$ is called 
a {\em permutation polynomial} of $\F$ if the function
 $ x \mapsto  F(x)$ is a permutation of $\F$.

 Boolean functions are also involved in this paper and are generally
of the form
\[
x\in\F~~\mapsto ~~Tr(P(x))\in\FB, 
\]
where $P$ is any function from \(\F\) into \(\F\) and
where $Tr$ denotes the {\em absolute trace} on $\F$, {\em i.e.,}
\[
Tr(\beta)=\beta+\beta^2+\dots+\beta^{2^{n-1}}, ~\beta\in\F.
\]

 In the whole paper, $\#E$ denotes the cardinality of any set \(E\).

The resistance of a cipher to differential attacks and to its variants
is quantified by some properties of the {\em derivatives} of its
S(ubstitution)-box, in the sense of the following definition.  It is
worth noticing that this definition is general: it deals with functions
from \(\FF_{2^n}\) into \(\FF_{2^m}\) for any \(m \geq 1\). 

\begin{definition}\label{de1}
Let \(F\) be a function from \(\F\) into \(\FF_{2^m}\). For
any \(a \in \F\), the {\em derivative of \(F\) with respect
  to \(a\)} is the function \(D_a F\) from \(\F\) into \(\FF_{2^m}\) defined by 
\[D_aF(x) = F(x+a) + F(x), \;\; \forall x \in \F.\] 
\end{definition}
The resistance to differential cryptanalysis is  related to the
following quantities, introduced by Nyberg and 
Knudsen~\cite{Nyberg_Knudsen92,Nyberg93}.
\begin{definition}\label{de2}
Let \(F\) be a function from \(\F\) into \(\F\). For
any \(a\) and \(b\) in \(\F\), we denote
\[\delta(a,b) = \# \{x \in \F, \; D_aF(x) = b\}.\]
 Then, the {\em differential uniformity} of $F$ is
\[\delta(F) = \max_{a \neq 0, \; b \in \F} \delta(a,b).\]
Those functions for which $\delta(F)=2$ are said to be {\em almost
perfect nonlinear (APN)}.
\end{definition}

\subsection{Differential spectrum of power functions}\label{sec-difpo}

In this paper, we focus on the case where  the S-box is a power function, 
{\em i.e.,} a
 monomial function on $\F$. In other words, $F(x)=x^d$ over \(\F\), which
 will be denoted by \(F_d\) when necessary.  
In the case of such a power function, the differential properties
 can be analyzed more easily since, for any nonzero \(a \in \F\), the
equation $(x+a)^d + x^d = b$ can be written
\[
a^d\left(\left(\frac{x}{a}+1\right)^d + \left(\frac{x}{a}\right)^d\right)=b,
\]
implying that 
\[
\delta(a,b)=\delta(1,b/a^d)~\mbox{for all \(a \neq 0\)}.
\]
 Then, when \(F: x \mapsto x^d\) is a monomial function, the differential characteristics of \(F\) are determined by the values 
\(\delta(1, b)\), \(b \in \F\). From now on, this quantity  $\delta(1,b)$
 is denoted by $\delta(b)$. 
Since
\[
\#\{b\in\F|\delta(a,b)=i\}=\#\{b\in\F|\delta(b)=i\}~~~
\mbox{ $\forall a\neq 0$},\]
the {\em   differential spectrum} of $F$ can be defined as follows.
\begin{definition}\label{de:specdiff}
Let $F(x)=x^d$ be a  power function on $\F$.  We denote by
$\omega_i$ the number of output differences~\(b\) that occur $i$~times:
\begin{equation}\label{deomega}
\omega_i=\#\{b\in\F|\delta(b)=i\}.
\end{equation}
The {\it differential spectrum} of~\(F_d\) is the set of $\omega_{i}$:
\begin{center}
\(\mathbb{S}=\{\omega_{0},\omega_{2},...,\omega_{\delta(F)}\}\).
\end{center}
\end{definition}
With same notation, we have the following equalities. They are 
well-known but we indicate the proof for clarity. 

\begin{lemma}\label{ident}
\[
\sum_{k=0}^{2^n} \omega_k = 2^n ~~\mbox{and}~~\sum_{k=2}^{2^n} 
(k\times  \omega_k) = 2^n,
\]
where $\omega_i=0$ for $i$ odd.
\end{lemma}
\begin{proof}
The first equality is obviously deduced from (\ref{deomega}). And,
for $k>0$, $k\times  \omega_k$ equals the number of $x\in\F$ such that
\[
x^d+(x+1)^d=b  ~~\mbox{and}~~\delta(b)=k
\]
for some $b$. Thus, any $x$ is counted in the second sum.
\end{proof}

\begin{remark}
The differential spectrum of the power function \(F(x)=x^d\) over \(\F\) is also related to the weight enumerator of the cyclic code of length~\((2^n-1)\) with defining set \(\{1, s\}\)~\cite{CCZ}. In particular, the number of codewords with Hamming weight~\(3\) and~\(4\) in this cyclic code can be derived from the differential spectrum of~\(F\) (see e.g. Corollary~1 in~\cite{BCC}).
\end{remark}

A power function $F$ is said to be {\em differentially $2$-valued}
 if and only if  for any $b\in\F$, we have
$\delta(b)\in\{0,\kappa\}$ (and then  only two $\omega_i$
in $\mathbb{S}$ do not vanish). It is known that $\kappa=2^r$ for some $r>1$
(see an extensive study in \cite[Section 5]{BCC}).
Note that APN functions
are differentially $2$-valued with $\kappa=2$.

There are some basic transformations which preserve $\mathbb{S}$.

\begin{lemma}
Let    $F_d(x)=x^d$   and    $F_e(x)=x^e$ over  $\F$. If there exists $k$
 such that  $ e=2^kd \bmod{2^n-1}$ or
if $ed = 1 \bmod{2^n-1}$, then $F_e$ has the same differential spectrum
as $F_d$.
\end{lemma}

\subsection{General properties on the differential spectrum}

In this section, $F_d(x)=x^d$ and notation is as in \S~\ref{sec-difpo}. 
Studying \(\delta(b) \) for special values of $b$ may give us as least
a lower bound on $\delta(F_d)$. So we first focus on~\(\delta(0)\). 
\begin{lemma}\label{delt0}
Let $d$ be such that $\gcd(d,2^n-1)=s$. Then  \(F_d:x \mapsto x^d\)
is such that $\delta(0)=s-1$.
In particular  $s=1$  if and only if $\delta(0)=0$.
\end{lemma}
\begin{proof}
Note that  $s=1$ if and only if $F_d$ is a permutation.
Obviously, $x$ is a solution of $x^d+(x+1)^d=0$ if and only
if
\[
\left(\frac{x+1}{x}\right)^d=1~\mbox{that is $x+1=xz$ with $z^d=1$,}
\]
since $x\mapsto (x+1)/x$ is a permutation over $\F\setminus\{0,1\}$.
As there are exactly $s-1$ such nonzero $z$, the proof is completed.
\end{proof}
There is an immediate consequence of Lemma \ref{delt0} for specific
values of $d$. 
\begin{proposition}\label{delprime}
Let  $d\geq 3$ such that $d$ divides $2^n-1$. Then
$\delta(F_d)=\delta(0)=d-1$.

In particular, if $d=2^t-1$ with $\gcd(t,n)=t$ then 
$\delta(F_d)=\delta(0)=2^t-2$. 
\end{proposition}
\begin{proof}
Since  $\gcd(d,2^n-1)=d$,  $\delta(0)= d-1$ from Lemma \ref{delt0}. 
But the polynomial
 $x^d+(x+1)^d+b$ has degree $d-1$ for any $b$, so that $\delta(b)\leq d-1$.
We conclude that $\delta(F_d)=d-1$.

Now, let $d=2^t-1$ with $\gcd(t,n)=t$. Then $\gcd(d,2^n-1)=2^t-1$ so
that $\delta(0)=2^t-2$. As previously we conclude that $\delta(F_d)=2^t-2$.
\end{proof}
\begin{ex}\label{ex1}
If $d=3$ then $\delta(F_d)=\delta(0)=2$ for any even $n$.\\
If $d=5$ then $\delta(F_d)=\delta(0)=4$ for   $n=4k$  for all
$k>1$.\\
If $d=7$ then $\delta(F_d)=\delta(0)=6$ for   $n=3k$  for all
$k>1$.
\end{ex}
The previous remarks combined with our simulation results point out that \(\delta(0)\) and \(\delta(1)\) play a very particular role in the differential spectra of power functions. This leads us to investigate the properties of the differential spectrum restricted to the values \(\delta(b)\) with \(b \not \in \FB\).
\begin{definition}\label{loc-apn}
Let $F$ be a  power function on $\F$. We say that $F$  has the 
{\em same restricted
 differential spectrum as an APN function} when
\[
\delta(b)\leq 2 \mbox{ for all } b \in \F \setminus \FB.
\]
For the sake of simplicity, we will say that $F$ is {\em locally-APN}.
\end{definition}

This definition obviously generalizes the APN property. For instance, 
the {\em inverse} function over $\F$ is locally-APN for
any $n$, while it is APN for odd \(n\) only. Another infinite class of 
locally-APN functions is exhibited in Section~\ref{sec-n/2}. 

\section{The differential spectrum of 
$\boldsymbol{x\mapsto x^{2^t-1}}$}\label{sec-ds}
From now on, we investigate the differential spectra of the following
 specific monomial functions 
$$G_t: x \mapsto x^{2^t-1}, ~2\leq t\leq n-1,~\mbox{~over~\(\F\)}\;.$$
Note that such a function has algebraic degree $t$.

\subsection{Link with linear polynomials}
In this section, we first give some general properties.

\begin{theorem}\label{th: 2^t-1}
Let $G_t(x)=x^{2^t-1}$ over $\F$ with $2\leq t\leq n-1$. Then, 
\begin{equation}\label{eq:u1}
 G_t(x+1) + G_t(x) + 1 = \frac{(x^{2^{t-1}}+x)^2}{x^2+x}.
\end{equation}
Consequently, for any \(b \in \F \setminus \{1\}\), \(\delta(b)\) is the
 number of roots in \(\F \setminus \FB\) of the linear polynomial
\[P_b (x) = x^{2^t} + b x^2 + (b+1) x\;.\]
And  we have
\begin{eqnarray*}
\delta(0) & = & 2^{\gcd(t,n)}-2 \\
\delta(1) & = & 2^{\gcd(t-1,n)} \\
\mbox{for any}~~ b \in \F \setminus \FB, \; \; \delta(b) & = & 2^r-2
\end{eqnarray*}
for some $r$ with $1\leq r \leq \min(t,n-t+1)$.
\end{theorem}
\begin{proof}
To prove (\ref{eq:u1}) we simply compute
\[
(x+x^2)(1+x^{2^t-1}+(1+x)^{2^t-1})=x+x^2+x^{2^t}+x^{2^t+1}+x(1+x)^{2^t}
= x^2+x^{2^t}.
\]

Thus, \(\delta(1)\) is directly deduced and it corresponds to the number
 of roots of \(P_1(x) = (x^{2^{t-1}} + x)^2\).
Let \(b \in \F \setminus \{1\}\). Then \(x \in \F \setminus \FB\) is a
solution of 
\[
(x+1)^d + x^d = b,~d=2^{t}-1,
\]
if and only if it is a solution of 
\[(x^{2^{t-1}} + x)^2 = (b+1) x (x+1),\]
or equivalently if it is a root of the linear polynomial
\[P_b(x) = x^{2^t} + bx^2 + (b+1)x.\]
The values \(x=0\) and \(x=1\) are counted in $\delta(1)$ (as
solutions of~\((x+1)^d+x^d=1\)), while $P_b(0) =P_b(1) =0$ for any $b$.
So, we get that, if \(b \neq 1\), the number of roots of~\(P_b\) in~\(\F\) is equal to~\((\delta(b)+2)\).
Because the set of all roots of a linear polynomial is a
linear space, we deduce that
\[
\forall b \in \F \setminus \{1\}, \; \; \delta(b) =
2^r-2 \mbox{ with } r \leq t.
\]
Moreover, by raising \(P_b\) to the \(2^{n-t}\)-th power, we get that any
root of~\(P_b\) is also a root of 
\[b^\prime x^{2^{n-t+1}} + (b^\prime + 1) x^{2^{n-t}} + x\]
with \(b^\prime = b^{2^{n-t}}\). This then implies that \(\delta(b) =
2^r -2 \) with \(r \leq n-t+1\).
Finally, for \(b=0\), \(P_0(x)=x^{2^t}+x\), implying that \(\delta(0)= 2^{\gcd(t,n)}-2\), which naturally corresponds to Lemma~\ref{delt0}.
\end{proof}
\begin{remark}\label{reminv}
As a first easy corollary, we recover the following well-known form of the 
differential spectrum of the inverse function, 
\(G_{n-1}: x \mapsto x^{2^{n-1}-1}\) over \(\F\). Actually, the previous 
theorem applied to \(t=n-1\) leads to \(\delta(0)=0\) and \(\delta(1)=2\)
 when \(n\) is odd and \(\delta(1)=4\) when \(n\) is even. For all 
\(b \not \in \FB\), \(\delta(b) \in \{0, 2\}\). Therefore, we have
\begin{itemize}
\item if \(n\) is odd, \(\delta(G_{n-1})=2\) and \(\omega_0=2^{n-1}\), \(\omega_2=2^{n-1}\);
\item if \(n\) is even, \(\delta(G_{n-1})=4\) and \(\omega_0=2^{n-1}+1\), 
\(\omega_2=2^{n-1}-2, ~\omega_4=1\).
\end{itemize}
Clearly $G_{n-1}$ is locally-APN for any $n$, as we previously noticed
(see Definition~\ref{loc-apn}).
\end{remark}

The following corollary is a direct consequence of Theorem~\ref{th: 2^t-1}.
\begin{corollary}\label{coro:delta}
Let $G_t(x)=x^{2^t-1}$ over $\F$ with $2\leq t\leq n-1$. Then, 
its differential uniformity is of the form either \(2^r-2\) or \(2^r\) for some \(2 \leq r \leq n\). Moreover, if \(\delta(G_t)=2^r\) for some \(r >1\), then this value appears only once in the differential spectrum, {\em i.e.}, \(\omega_{2^r}=1\), and it corresponds to the value of \(\delta(1)\), implying \(\delta(G_t)=2^{\gcd(t-1,n)}\).
\end{corollary}

\subsection{Equivalent formulations}
In Theorem \ref{th: 2^t-1}, we exhibited some tools for the computation
of the diffe\-rential spectra of functions $x\mapsto x^{2^t-1}$.
The problem boils down {\em to the determination of the roots of a linear
polynomial} whose coefficients depend on $b\in\F$.
There are equivalent formulations that we are going to develop
now. The first one is obtained by introducing another class of linear polynomials
over $\F$. For any subspace $E$ of $\F$ (where $\F$ is identified with $\FB^n$), we define its {\em dual}
as follows:
\[
E^\perp=\{~x~|~Tr(xy)=0,~\forall~y\in E~\}.
\]
Also, we denote by \(\mathcal{I}m(F)\) the image set of any function $F$.
\begin{lemma}\label{le:sym}
Let $t,s \geq 2$ and $s=n-t+1$. 
Let us consider the linear applications
\[
P_{t,b}(x)=x^{2^t}+bx^2+(b+1)x,~~b\in\F.
\]
Then the dual of  $\mathcal{I}m(P_{t,b})$ is the set of all \(\alpha\)
satisfying $P^*_{t,b}(\alpha)=0$ where 
\[
P^*_{t,b}(x)=x^{2^s} +(b+1)^2x^2+bx.
\]
Note that $P^*_{t,b}$ is called the  {\em adjoint} 
application of $P_{t,b}$.
\end{lemma}

\begin{proof}
By definition, \(\mathcal{I}m(P_{t,b})^\perp\) consists of all $\alpha$ such that  $Tr(\alpha P_{t,b}(x))=0$ for all $x \in \F$. We have
\begin{eqnarray*}
Tr(\alpha P_{t,b}(x)) &=& 
Tr(\alpha x^{2^t})+Tr(b\alpha x^2)+Tr(\alpha(b+1)x)\\ &=&
Tr(\alpha^{2^{n-t+1}} x^2)+Tr(b\alpha x^2)+Tr(\alpha^2(b+1)^2x^2)\\ &=&
Tr(x^2(\alpha^{2^s} +\alpha^2(b+1)^2+\alpha b)).
\end{eqnarray*}

Hence \(\alpha\) belongs to the dual of the image of $P_{t,b}$ if and
only if $\alpha^{2^s} +\alpha^2(b+1)^2+\alpha b=0$, {\em i.e.},
$\alpha$ is a root of $P^*_{t,b}$, completing the proof.
\end{proof}
The following theorem gives an equivalent formulation of the quantity $r$
which is presented in Theorem \ref{th: 2^t-1}.
\begin{theorem}\label{thadj}
Notation is as in Lemma \ref{le:sym}. Then
\[
\dim Ker(P_{t,b})=\dim Ker(P^*_{t,b}).
\]
Consequently, this dimension can be determined by solving $P_{t,b}(x)=0$
or equivalently by solving
\[x^{2^s} +(b+1)^2x^2+bx=0, ~\mbox{ where } s=n-t+1.\]
\end{theorem}
\begin{proof}
Let $\kappa$ be the dimension of the image set of $P_{t,b}$.
It is well-known that $n=\kappa+\dim Ker(P_{t,b})$.
On the other hand, Lemma~\ref{le:sym} shows that $\alpha$ is in the dual of the image  of $P_{t,b}$ if and only if 
$P^*_{t,b}(\alpha)=0$. We deduce that 
\[
n-\kappa = \dim Ker(P^*_{t,b})=\dim Ker(P_{t,b})\;,
\]
completing the proof.
\end{proof}
Now, we discuss a different point of view, using an equivalent linear system.
\begin{theorem}\label{thsys}
For any $2\leq t<n$, we define the following equations:
\[
E_b~:~x^{2^t} + b x^2 + (b+1) x=0, ~b\in \F.
\]
Let $N_b$ be the number  of solutions of $E_b$ in $\F\setminus\FB$.
Let $M_b$ be the  number  of solutions in $\F^*$ of the system
$$
\left.
\begin{array}{rcc}
y^{2^{t-1}}+\dots +y^2+y(b+1) &=& 0\\
Tr(y) &=& 0
\end{array}
\right\}
$$
Then $N_b=2\times M_b$.
\end{theorem}
\begin{proof}
We simply write
\[x^{2^t} + b x^2 + (b+1) x = x^{2^t}+x+b(x^2+x)\]
which is equal to
\begin{eqnarray*}
&=& (x^2+x)^{2^{t-1}}+(x^2+x)^{2^{t-2}}+\cdots + (x^2+x)+b(x^2+x)\\
&=& y^{2^{t-1}}+y^{2^{t-2}}+\cdots y^2+y(b+1),~\mbox{with $y=x^2+x$}.
\end{eqnarray*}
We are looking at the number of solutions of $E_b$ which are not
in $\FB$. So, it is equivalent to compute the number of nonzero solutions
$y$ of 
\[y^{2^{t-1}}+y^{2^{t-2}}+\cdots y^2+y(b+1)=0\] 
such that the equation $x^2+x+y=0$ has
solutions. This last condition holds if and only if $Tr(y)=0$, providing 
two distinct solutions $x_1,x_2=x_1+1$ such that $x_i^2+x_i=y$, completing 
the proof.
\end{proof}
\begin{remark}\label{val-b}
In Theorem~\ref{thsys}, $b$ takes any value while $P_b$ is defined
for $b\neq 1$ in  Theorem~\ref{th: 2^t-1}. For all $b\neq 1$, we have 
clearly $N_b=\delta(b)$. If \(b=1\), \(P_1(x)=x^{2^t} + x^2\) and the number
 of roots of~\(P_1\) in \(\F\) is equal to 
\[
N_1+2=2^{\gcd(t-1,n)}=\delta(1).
\]
Therefore, we have \(M_1 = \delta(1)/2 - 1\).
\end{remark}
\section{A property of symmetry}\label{sec-sym}
Recall that  $G_t(x)=x^{2^t-1}$.
Now, we are going to examine some symmetries between the differential 
spectra of $G_t$ and $G_s$ where $t,s \geq 2$ and $s=n-t+1$. 
In the list of properties below, notation is conserved 
as soon it is defined. 
Recall that 
\[
P^*_{t,b}(x)=x^{2^s} +x^2(b+1)^2+xb
\]
is the adjoint polynomial of $P_{t,b}(x)=x^{2^t}+bx^2+(b+1)x$. 
Thus, both polynomials have a kernel with the same dimension
(see Lemma \ref{le:sym} and Theorem~\ref{thadj}). 
It is worth noticing that  this dimension is at least $1$
since $P_{t,b}(0)=P_{t,b}(1)=0$.
In this section we want to prove the following theorem.

\begin{theorem}\label{thmain}
For any $\nu$ with $2\leq\nu\leq n-1$, we define
\[
S^i_\nu=\{~b~|~\dim Ker(P_{\nu,b})=i~\} \mbox{ with }  1\leq i\leq \nu\;.
\]
Then, for any $s,t\geq 2$ with $t=n-s+1$ and for any $i$,
we have $\# S^i_s=\# S^i_t$.  
\end{theorem}
We begin by some lemmas. The next one will not be used for
the proof of Theorem \ref{thmain} but clarifies some arguments.
\begin{lemma}\label{le:rac}
Let $a\in\F^*$ and $2\leq t\leq n-1$. Then there are exactly
 two elements,
$b_1$ and $b_2$ with $b_2=b_1+a^{-1}$, such that $P^*_{t,b_i}(a)=0$ for $i=1,2$.
In particular, $P^*_{t,b}(1)=0$ for $b\in\{0,1\}$.
\end{lemma}
\begin{proof}
Let $a$ be fixed and let us consider the equation $P^*_{t,b}(a)=0$ for some $b$:
\[
b^2a^2+ba+a^{2^s}+a^2=a^2\left(b^2+\frac{b}{a}+\frac{a^{2^s}+a^2}
{a^2}\right)=0.
\]
There is $b$ such that this equation is satisfied if and only if 
\[
Tr\left(\frac{a^{2^s}+a^2}{a^2}\times a^2\right)=Tr(a^{2^s}+a^2)=0,
\]
which holds for any $a$. Thus, for any nonzero $a$ there are exactly 
two solutions, say $b_1$ and $b_2$ whose sum equals $a^{-1}$.
To complete the proof, we observe that $P^*_{t,b}(1)=b^2+b$.
\end{proof}
\begin{lemma}\label{le:eqv}
Let $s,t\geq 2$ with $t=n-s+1$. Let \(\pi\) be the permutation of \(\F^* \times \F\) defined by
\[\pi(a,b) = \left(a^{2^s}, \frac{ab}{a^{2^s}}+1\right)\;.\]
Then, for any \((a,b)\) in \(\F^* \times \F\), \((\alpha, \beta)=\pi(a,b)\) satisfies
\[P_{s,\beta}^*(\alpha) = P_{t,b}^*(a)\;.\]
\end{lemma}

\begin{proof}
First, we clearly have that \(\pi\) is a permutation of \(\F^* \times \F\).
Indeed,  \(\pi\left(\F^* \times \F\right) \subset \F^* \times \F\) and one
can define the inverse of $\pi$ as follows:
\[\pi^{-1}(\alpha, \beta) = \left(\alpha^{2^{n-s}}, \frac{\alpha(\beta+1))}{
\alpha^{2^{n-s}}}\right)\;.
\]
Actually, $(\alpha^{2^{n-s}})^{2^s}=\alpha$ and it can be checked that
\[
\pi(\pi^{-1}(\alpha, \beta)) = \left(\alpha,\frac{\alpha^{2^{n-s}}\alpha(\beta+1)}{\alpha\alpha^{2^{n-s}}}+1\right)=(\alpha,\beta).
\]
Then, by using that
\((\beta+1)^2=\frac{a^2b^2}{a^{2^{s+1}}}\) and $s+t=n+1$,
we deduce that
\begin{eqnarray*}
P_{s,\beta}^*(\alpha) & = & (a^{2^s})^{2^t}+(a^{2^s})^{2}(\beta+1)^2+
(a^{2^s})\beta \\
& = & a^2 + a^2b^2 + ab + a^{2^s}\\
& = & P_{t,b}^*(a)\;.
\end{eqnarray*}
\end{proof}
\begin{lemma}\label{le:eq1}
Let $s,t\geq 2$ with $t=n-s+1$. Let $b\in\F$
and let $a\in\F^*$ such that 
$P^*_{t,b}(a)=0$.
 Then $\dim Ker(P^*_{t,b})=\dim Ker(P^*_{s,\beta})$,
where $\beta=1+ab/a^{2^s}$.
\end{lemma}
\begin{proof}
Recall that $P^*_{t,b}(x)=x^{2^s} +x^2(b+1)^2+xb$. We know that for
any $b\not\in\FB$ there is $a\in\F\setminus\{0,1\}$ such that
$P^*_{t,b}(a)=0$. This is because $\dim Ker(P_{t,b})=\dim
Ker(P^*_{t,b})$ (see Theorem~\ref{thadj}) and $\{0,1\}$ is included in
the kernel of $P_{t,b}$. Moreover, $P^*_{t,b}(1) = b^2+b= 0$ if and only if
 $b \in \FB$.

We treat the case $a=1$ separately, a case where $P_{t,b}(a)=0$ for
$b\in\FB$ only. In this case, Lemma \ref{le:eqv} leads to $P^*_{s,\beta}(1)=0$ 
too where $\beta=b+1$, since $\pi(1,b)=(1,b+1)$. And we have for $b=0$
\[
P^*_{t,0}(x)=x^{2^s}+x^2=P_{s,1}(x)
\]
and for $b=1$
\[
P^*_{t,1}(x)=x^{2^s}+x=P_{s,0}(x).
\]
Thus, we conclude: for $a=1$, if \(b\) is such that  $P^*_{t,b}(1)=0$ then
$\beta=b+1$ and 
$$\dim Ker(P^*_{t,b})=\dim Ker(P_{s,\beta})=\dim Ker(P^*_{s,\beta})$$ 
where the last equality comes from Theorem \ref{thadj}.

Now, we suppose that $a\not\in\FB$. With $x=ay$, the equation 
$P^*_{t,b}(x)=0$  is equivalent to
\[
a^{2^s}y^{2^s}+a^2y^2(b+1)^2+ayb=0
\]
which is
\[
a^{2^s}\left(y^{2^s}+\frac{a^2(b+1)^2}{a^{2^s}}y^2+
y\frac{ab}{a^{2^s}}\right)=0.
\]
We can set 
\[
\beta=\frac{a^2(b+1)^2}{a^{2^s}}~~\mbox{and}~~
\beta+1=\frac{ab}{a^{2^s}},
\]
 since
\[
\frac{a^2(b+1)^2}{a^{2^s}}+1=\frac{ab}{a^{2^s}}
\]
is equivalent to
\[
a^{2^s}+a^2(b+1)^2+ab=0,~ {\it i.e.,}~P^*_{t,b}(a)=0.
\]
We have proved that $P^*_{t,b}(x)=0$  is equivalent to 
\[
P_{s,\beta}(y)=y^{2^s}+\beta y^2+(\beta+1)y=0.
\]
Then, $\dim Ker(P_{s,\beta})=\dim Ker(P^*_{t,b})$.
But $\dim Ker(P_{s,\beta})=\dim Ker(P^*_{s,\beta})$ by Theorem \ref{thadj},
completing the proof.
\end{proof}

\bigskip
\noindent
{\em Proof of Theorem  \ref{thmain}}. 
Recall that 
\[S^i_\nu=\{~b\in\F~|~\dim Ker(P_{\nu,b})=i~\}\;.\]
Then, we want to show that, for any \(i\), \(\#S^i_t = \#S^i_s\).
For any $2\leq\nu\leq n-1$ and for any $1\leq i\leq \nu$, we define
\[
\mathcal{E}^i_\nu=\{(a,b) \in\F^*  \times \F~|~P^*_{\nu,b}(a)=0 \mbox{ and }
\dim Ker(P_{\nu,b})=i~\}\;.
\]
From Theorem~\ref{thadj}, we know that \(\dim Ker(P_{\nu,b})=\dim Ker(P^*_{\nu,b})\). Then,
\[\mathcal{E}^i_\nu=\{(a,b) \in\F^* \times \F~|~P^*_{\nu,b}(a)=0 \mbox{ and }\dim Ker(P^*_{\nu,b})=i~\}.
\]
For any $b\in S^i_\nu$  there are $2^i-1$ nonzero $a$ in  $Ker(P^*_{\nu,b})$
and then $2^i-1$ pairs $(a,b)$, for a fixed $b$, in $\mathcal{E}^i_\nu$ so that
\begin{equation}\label{eq0}
\# \mathcal{E}^i_\nu = (2^i-1) \#S^i_\nu\;.
\end{equation}
We use Lemma \ref{le:eqv}. Recall that \(\pi\) is the permutation of 
\(\F^*  \times \F\) defined by
\[\pi(a,b) = \left(a^{2^s}, \frac{ab}{a^{2^s}}+1\right).
\]
Then, we have
\begin{eqnarray*}
\mathcal{E}^i_t & = &  \{(a,b)  \in\F^* \times \F~|~P^*_{t,b}(a)=0 
\mbox{ and }\dim Ker(P^*_{t,b})=i~\},\\
\mathcal{E}^i_s & = & \{(\alpha,\beta) \in\F^* \times \F~|~
P^*_{s,\beta}(\alpha)=0 \mbox{ and }\dim Ker(P^*_{s,\beta})=i~\} \\
& = & \{(\alpha,\beta) = \pi(a,b), (a,b) \in \mathcal{E}^i_t\}\;.\\
\end{eqnarray*}

Indeed, any $(\alpha,\beta)$ is as follows specified from $(a,b)$.
We have $P^*_{s,\beta}(\alpha)=P^*_{t,b}(a)$ from Lemma \ref{le:eqv}.
Moreover, according to Lemma \ref{le:eq1}, 
$\dim Ker(P^*_{t,b})=\dim Ker(P^*_{s,\beta})$, where $\beta$
is calculated from $a$ and $b$, for any $a$ such that 
$P^*_{t,b}(a)=0$. 

In other terms, to any pair $(a,b)\in\mathcal{E}^i_t$ corresponds
a unique pair $(\alpha,\beta)\in\mathcal{E}^i_s$.
We finally get that $\# \mathcal{E}^i_s =\# \mathcal{E}^i_t$
and it directly follows from~(\ref{eq0}) that
\(\#S^i_s = \#S^i_t\), completing the proof.
\hfill$\diamond$\\

Now we are going to explain Theorem \ref{thmain}, 
in terms of the differential 
spectra of $G_t$ and $G_s$, $s,t\geq 2$ with $t=n-s+1$. 
Actually, we can deduce from the previous theorem that both functions \(G_t\)
 and \(G_s\) have the same {\em restricted differential spectrum}, 
{\em i.e.} the multisets \(\{\delta(b), b \in \F \setminus \FB\}\) are
 the same for both functions.
\begin{corollary}\label{cormain}
We denote by $\delta_\nu(b)$, $b\in\F$, 
the quantities $\delta(b)$ corresponding to $G_\nu: x \mapsto x^{2^\nu-1}$.
Then, for any $s,t\geq 2$ with $t=n-s+1$, we have
\[\begin{array}{lclcl}
\delta_s(0) & = & \delta_t(1)-2 & = & 2^{\gcd(t-1,n)}-2 \\
\delta_s(1) & = & \delta_t(0)+2 & = & 2^{\gcd(t,n)}
\end{array}\]
and we have equality between both following multisets:
\begin{equation}\label{msets}
\{\delta_s(b), b \in \F \setminus \FB\} = \{\delta_t(b), b \in \F 
\setminus \FB\}.
\end{equation}
Moreover, \(G_t\) and \(G_s\) have the same differential spectrum
 if and only if 
\[
\gcd(s,n)=\gcd(t,n)=1,
\]
which can hold for odd $n$  only.
In any case, $G_t$ is locally-APN if and only if $G_s$ is  locally-APN.
\end{corollary}
\begin{proof}
Since $s=n-t+1$, we clearly have
\[
\gcd(s,n)=\gcd(t-1,n)~~\mbox{and}~~\gcd(s-1,n)=\gcd(t,n).
\]
Thus, applying Theorem \ref{th: 2^t-1}, we get
\[
\delta_s(0)=2^{\gcd(s,n)}-2=2^{\gcd(t-1,n)}-2=\delta_t(1)-2
\]
and
\[
\delta_s(1)=2^{\gcd(s-1,n)}=2^{\gcd(t,n)}=\delta_t(0)+2.
\]
Moreover, we have
\begin{eqnarray*}
\left(P_{t,1}(x)\right)^{2^{s-1}} & = & \left(x^{2^t} + x^2\right)^{2^{s-1}} = x + x^{2^s} = 
P_{s,0}(x)\\
\left(P_{t,0}(x)\right)^{2^s} & = & \left(x^{2^t} + x \right)^{2^s} = x^{2^{s}}  + x^2 =
 P_{s,1}(x)\;,
\end{eqnarray*}
implying that 
\[
\{\dim Ker P_{t,0} , \dim Ker P_{t,1}\} = \{\dim Ker P_{s,0} , \dim Ker P_{s,1}\}\;.\]
We deduce from Theorem~\ref{thmain} that
\[\# \{~b\in\F \setminus \FB ~|~\dim Ker(P_{t,b})=i~\} = \# \{~b\in\F \setminus \FB ~|~\dim Ker(P_{s,b})=i~\}.\]
Equality (\ref{msets}) is then a direct consequence of Theorem~\ref{th: 2^t-1}, since
\[\{\delta_{\nu}(b), b \in \F \setminus \FB\} = \{2^{\kappa(b)} -2, \; \kappa(b) = \dim Ker(P_{\nu,b})\}\;.\]

Now,  we note that 
$\delta_s(0)=\delta_t(0)$ if and only if $\delta_s(1)=\delta_t(1)$.
Thus, \(G_t\) and \(G_s\) have the same differential spectrum if and only if $\delta_s(0)=\delta_t(0)$.
Since
\[
\delta_s(0)=2^{\gcd(s,n)}-2~~\mbox{and}~~\delta_t(0)=2^{\gcd(t,n)}-2,
\]
this holds if and only if  $\gcd(t,n)=\gcd(s,n)=1$. It cannot hold
when $n$ is even, because in this case either $s$ or $t$ is even
too.

Using Definition \ref{loc-apn},  the last statement is obviously derived.
\end{proof}

The previous result implies that, if \(G_t\) is APN over \(\F\), then \(G_s\)
 is locally-APN. Moreover, the differential spectrum of~\(G_s\) can be completely determined as shown by the following corollary.
\begin{corollary}\label{cr:apn}
Let \(n\) and \(t <n\) be two integers such that \(G_t: x \mapsto x^{2^t-1}\) is APN over \(\F\). Let \(s = n-t+1\). Then,
\begin{itemize}
\item if \(n\) is odd, both \(G_t\) and \(G_s\) are APN permutations;
\item if \(n\) is even, \(G_t\) is not a permutation and \(G_s\) is a
 differentially \(4\)-uniform permutation (locally-APN) with the following
 differential spectrum:
\(\omega_4=1\), 
\(\omega_2=2^{n-1}-2\) and \(\omega_0=2^{n-1}+1\).
\end{itemize}
\end{corollary}

\begin{proof}
From Theorem~\ref{th: 2^t-1}, we deduce that, if \(F\) is APN, then
 \(\delta_t(b) \in \{0,2\}\) for all \(b \in \F \setminus \FB\); moreover, 
\(\gcd(n,t-1)=1\) and \(\gcd(n,t) \in \{1,2\}\) since \(\delta_t(1)=2\)
 and \(\delta_t(0) \in \{0,2\}\).

If \(n\) is odd, \(\gcd(n,t)=1\) is then the only possible value, implying that \(\delta_t(0)=0\). It follows that \(\delta_s(0)=0\), \(\delta_s(1)=2\) and \(\delta_s(b) \in \{0,2\}\) for all \(b \in \F \setminus \FB\). 
In other words, both \(G_t\) and \(G_s\) are APN permutations.

If \(n\) is even, it is well-known that \(G_t\) is not a permutation (see e.g. \cite{BCCL}). More precisely, we have here \(\gcd(n,t)=2\) since \(t\) and \(t-1\) cannot be both coprime with~\(n\). Then, we deduce that \(\delta_s(0)=0\) and \(\delta_s(1)=4\). The differential spectrum of \(G_s\) directly follows from Corollary~\ref{cormain}.
\end{proof}
\begin{ex}
Notation is as in Corollary \ref{cr:apn}.
For $t=2$, we have $G_t(x)=x^3$. It is well-known that $G_2$ is an APN
function over $\F$ for any $n$. Since $s=n-1$, $G_s(x)$ is equivalent to the inverse function and it is 
also well-known that the inverse function is APN for odd $n$.
For even $n$, $\delta(G_{n-1})=4$ and the differential spectrum 
is computed in Remark \ref{reminv}.
\end{ex}
\begin{corollary}\label{coro3}
Let \(n\) and \(t <n\) be two integers such that \(G_t: x \mapsto x^{2^t-1}\) is differentially \(4\)-uniform. Then, \(n\) is even and \(G_t\) is a permutation with the following differential spectrum:
\(\omega_4=1\), 
\(\omega_2=2^{n-1}-2\) and \(\omega_0=2^{n-1}+1\). 
Moreover, for \(s = n-t+1\), \(G_s\) is APN.
\end{corollary}
\begin{proof}
From Corollary~\ref{coro:delta}, we deduce that \(\delta(G_t)=4\) implies \(\gcd(n,t-1)=2\) and \(\omega_4=1\). In particular, \(n\) is even. Since \( \gcd(n,t-1)\) and \(\gcd(n,t)\) cannot be both equal to~\(2\), we also deduce that that \(G_t\) is a permutation. Its differential spectrum is then derived from Lemma~\ref{ident}.

Moreover, we have \(\delta_s(0)= 2\) and \(\delta_s(1)=2\), implying that \(G_s\) is APN.
\end{proof}
\section{Specific classes}\label{sec-spec}

In this section, we apply the results of Section \ref{sec-ds}
 to the study of the differential spectrum of
 \(G_t: x \mapsto x^{2^t-1}\), for special  values of $t$.

\subsection{The function $\boldsymbol{x\mapsto x^7}$}
We first focus on \(G_3: x \mapsto x^7\) over \(\F\), {\em i.e.}, \(t=3\).
In this case, we  determine the complete differential spectrum of the function.
Moreover, thanks to the work of Carlitz \cite{Carlitz69}, we emphasize that
this spectrum is related to some Kloosterman sums defined as follows.

\begin{proposition}\label{prop: kloos}\cite[Formula~(6.8)]{Carlitz69}
Let \(K(1)\) be the Kloosterman sum
\[K(1) = \sum_{x \in \F}(-1)^{Tr(x^{-1}+x)}\]
extended to \(0\) assuming that \((-1)^{Tr(x^{-1})}=1\) for \(x=0\).
Then, 
\[K(1) =  1+ \frac{(-1)^{n-1}}{2^{n-1}} \sum_{i=0}^{\lfloor
  \frac{n}{2}\rfloor} (-1)^{i}\binom{n}{2i} 7^i.\]
\end{proposition}

\begin{theorem}\label{th:x^7}
Let \(G_3: x \mapsto x^7\) over \(\F\) with \(n \geq 4\). Then, its differential spectrum is given by:
\begin{itemize}
\item if \(n\) is odd,
\begin{eqnarray*}
\omega_6 & = & \frac{2^{n-2}+1}{6} - \frac{K(1)}{8} \\
\omega_4 & = & 0 \\
\omega_2 & = & 2^{n-1}-3 \omega_6 \\
\omega_0 & = & 2^{n-1} + 2 \omega_6;
\end{eqnarray*}

\item if \(n\) is even,
\begin{eqnarray*}
\omega_6 & = & \frac{2^{n-2}-4}{6} + \frac{K(1)}{8} \\
\omega_4 & = & 1 \\
\omega_2 & = & 2^{n-1}-3 \omega_6 -2 \\
\omega_0 & = & 2^{n-1} + 2 \omega_6 +1.
\end{eqnarray*}
\end{itemize}
where \(K(1)\) is the Kloosterman sum defined as in Proposition~\ref{prop: kloos}. In particular, \(G_3\) is 
differentially \(6\)-uniform for all \(n \geq 6\).
\end{theorem}
To prove this theorem,  we need some preliminary results.
We first recall some basic results on cubic equations.
\begin{lemma}\label{cubic} {\rm \cite{BerRumSol67}}
The cubic equation $x^3+ax+b=0$, where $a \in \F$  and $b \in \F^*$
has a unique solution in $\F$ if and only if $Tr(a^3/b^2) \neq Tr(1)$.
In particular, if it has three distinct roots in $\F$, then
$Tr(a^3/b^2)=Tr(1)$.
\end{lemma}
\begin{proposition}\label{cubeq}{\rm \cite[Appendix]{KuHeCaHa96}}
Let $f_a(x)=x^3+x+a$ and 
\[
M_i=\#\{~a\in\F^*~|~\mbox{$f_a(x)=0$ has precisely $i$ solutions in $\F$}~\}.
\]
Then, we have for odd $n$
\[
M_0=\frac{2^n+1}{3},~M_1=2^{n-1}-1,~M_3=\frac{2^{n-1}-1}{3}
\]
and for even $n$
\[
M_0=\frac{2^n-1}{3},~M_1=2^{n-1},~M_3=\frac{2^{n-1}-2}{3}.
\]
\end{proposition}
Now we are going  to solve the equations $P_b(x)=0$
(see Theorem \ref{th: 2^t-1})
by solving a system of equations, including a cubic
equation, thanks to the equivalence presented in Theorem~\ref{thsys}.

\begin{theorem}\label{th:cubeq2}
Let 
\[
P_b(x)= x^8 + b x^2 + (b+1) x, ~b\in \F\setminus\{1\}
\]
The number $\nu_0$ of $b \in \F\setminus\{1\}$ such that
$P_b$ has no  roots in $\F\setminus\{0,1\}$ is given by
\[\nu_0 = \frac{2^n + (-1)^{n+1}}{3} + 2^{n-2} + (-1)^n
\frac{K(1)}{4}\]
where \(K(1)\) is the Kloosterman sum defined as in 
Proposition~\ref{prop: kloos}.
\end{theorem}
\begin{proof}
Let \(b \in \F \setminus \{1\}\).
According to Theorem \ref{thsys}  we know that the number (denoted by $N_b$)
of  roots in  \(\F\setminus \FB\) of 
$P_b$  is twice the number of roots in \(\F^*\) of the following system
where $\beta=b+1$:
\begin{eqnarray}\label{eqsys7}
\left\{
\begin{array}{rcl}
Q_\beta(y)&=& y^3+y+\beta = 0\\
Tr(y) &=& 0.
\end{array}
\right.
\end{eqnarray}
Since $\beta\neq 0$,  $Q_\beta(y)\neq 0$ for $y\in\FB$.
Then, for any \(\beta \neq 0\), the following situations may occur:
\begin{itemize}
\item \(Q_{\beta}\) has no root in \(\F\). In this case,
  \(N_b =0\).
\item \(Q_{\beta}\) has a unique root \(y \in \F\). From
  Lemma~\ref{cubic}, this
  occurs if and only if \(Tr(\beta^{-1}) \neq Tr(1)\). In this case,
  \(N_b=0\) if \(Tr(y)=1\) and \(N_b=2\) if \(Tr(y)=0\).
\item \(Q_\beta\) has three roots \(y_1, y_2, y_3 \in \F\).  
Since these roots are roots of a linear polynomial of 
degree $4$ then
$y_3=y_1+y_2$, implying $Tr(y_3)=Tr(y_1)+Tr(y_2)$. Then, at least one
  \(y_i\) is such that \(Tr(y_i)=0\). It follows that, in this case,
  \(N_b\) is either~\(6\) or~\(2\).
\end{itemize}
Let us now define 
\[B = \#\{\beta \in \F^*, \;\; Q_\beta \mbox{ has a unique root } y \in \F \mbox{ and } Tr(y) = 1\}.\]
From the previous discussion, we have
\begin{eqnarray*}
\nu_0 & = & \#\{\beta \in \F^*, \;\; Q_\beta \mbox{ has
  no root in } \F\} + B\\
& = & \frac{2^n + (-1)^{n+1}}{3} + B
\end{eqnarray*}
where the last equality comes from Proposition~\ref{cubeq}.
Let us now compute the value of \(B\). 
\begin{eqnarray*}
B & = & \#\{\beta \in \F^*, \;\; Q_\beta \mbox{ has a unique root } y \in \F \mbox{ and } Tr(y) = 1\} \\
& = & \#\{(y^3+y) \in \F^*, \;\;  Tr\left(\frac{1}{y^3+y}\right) \neq Tr(1) \mbox{ and } Tr(y) = 1\}\;,
\end{eqnarray*}
by using that \(\beta = y^3+y\).
But, we have
\[
\frac{1}{y^3+y}=\frac{1+y^2}{y^3+y}+\frac{y^2+y}{y^3+y}
+\frac{y}{y^3+y}=\frac{1}{y}+\frac{1}{y+1}+\frac{1}{y^2+1},
\]
implying that 
\[Tr\left(\frac{1}{y^3+y}\right) = Tr\left(\frac{1}{y}\right).\]
Therefore, 
\[B = \#\{(y^3+y) \in \F^*, \;\; Tr\left(\frac{1}{y}\right)
\neq Tr(1) \mbox{ and } Tr(y) = 1\}.\]
Now, we clearly have that \((y^3+y) =0\) if and only if \(y \in \FB\).
Moreover, two distinct elements \(y_1\) and \(y_2\) in \(\F \setminus \FB\) with \(Tr(y_1^{-1}) \neq Tr(1)\) and \(Tr(y_2^{-1}) \neq Tr(1)\) satisfy \((y_1^3+y_1) \neq (y_2^3+y_2)\) (otherwise, \(Q_\beta\) with \(\beta = y_1^3+y_1\) has at least \(2\)~roots in \(\F\)). 
Therefore, we deduce that
\[B = \#\{y \in \F \setminus \FB, \;\; Tr\left(\frac{1}{y}\right)
\neq Tr(1) \mbox{ and } Tr(y) = 1\}.\]

If \(n\) is odd, we deduce that
\[B = \#\{y \in \F \setminus \FB, \;\; Tr\left(\frac{1}{y}\right)
=0 \mbox{ and } Tr(y) = 1\}.\]
If \(n\) is even, we deduce that
\begin{eqnarray*}
B & = & \#\{y \in \F \setminus \FB, \;\; Tr\left(\frac{1}{y}\right)
=1 \mbox{ and } Tr(y) = 1\}\\
& = & \#\{y \in \F \setminus \FB,  \;\;Tr(y) = 1\} \\
& & \quad  - \#\{y \in
\F \setminus \FB, \;\; Tr\left(\frac{1}{y}\right) 
=0 \mbox{ and } Tr(y) = 1\} \\
& = & 2^{n-1} - \#\{y \in\F \setminus \FB, \;\; Tr\left(\frac{1}{y}\right) 
=0 \mbox{ and } Tr(y) = 1\}\;.
\end{eqnarray*}
On the other hand, by definition of the Kloosterman sum~\(K(1)\), we
have
\begin{eqnarray*}
K(1) -2 & = & \sum_{x \in \F \setminus \FB} (-1)^{Tr(x^{-1}+x)} \\
& = & -2 \#\{x \in \F \setminus \FB, Tr(x^{-1}+x)=1\} + 2^{n}-2\\
& = & -4 \#\{x \in \F \setminus \FB , Tr(x^{-1})=0 \mbox{ and }
Tr(x)=1\} + 2^{n}-2.
\end{eqnarray*}
Thus,
\[\#\{x \in \F \setminus \FB , Tr(x^{-1})=0 \mbox{ and }
Tr(x)=1\} =  2^{n-2} - \frac{K(1)}{4}.\]
We then deduce that, for any \(n\),
\[B = 2^{n-2} +(-1)^n \frac{K(1)}{4}.\]
It follows that
\[\nu_0 = \frac{2^n + (-1)^{n+1}}{3} +2^{n-2} +(-1)^n
\frac{K(1)}{4}.\]
\end{proof}
\begin{proof} (Proof of Theorem \ref{th:x^7})
 In accordance with Lemma \ref{ident}, we  obtain the differential 
spectrum of $G_3$ as soon as we are able to solve the 
following system:
\begin{equation}\label{syseq7}
\begin{array}{c}
\omega_0+\omega_2+\omega_4+\omega_6=2^n\\
2\omega_2+4\omega_4+6\omega_6=2^n
\end{array}
\end{equation}
Now,  we apply Theorem \ref{th: 2^t-1} and we recall first that
$\delta(b)\in\{0,2,6\}$ for any $b\in\F\setminus\{1\}$.
Moreover, we know that $\omega_0=\nu_0$  as defined in Theorem \ref{th:cubeq2}.

Since $t=3$,  $\gcd(t-1,n)$ equals $1$ for odd $n$ and $2$ otherwise.
Then, if $n$ is even then $\delta(1)=4$ else $\delta(1)=2$.
Thus, $\omega_4=1$ for even $n$ and $\omega_4=0$ otherwise.
From the second equation of  (\ref{syseq7}), we get
\[
\omega_2= 2^{n-1}-3\omega_6-2 \omega_4\]
and using the first equation of  (\ref{syseq7})
\[
\omega_6= 2^n - \omega_0 - \omega_2 - \omega_4 = 2^{n-1} - \omega_0  + \omega_4 + 3\omega_6\;,
\]
leading to
\[\omega_6=-2^{n-2}+\frac{\omega_0-  \omega_4}{2}\;.\]
Finally, we deduce from Theorem~\ref{th:cubeq2} that, for odd $n$,
\begin{eqnarray*}
\omega_6 &=& -2^{n-2}+\frac{\omega_0}{2}=-2^{n-3}+\frac{2^n + 1}
{6}- \frac{K(1)}{8}\\
&=& \frac{2^{n-2} + 1}{6}-\frac{K(1)}{8}
\end{eqnarray*}
and for even $n$
\begin{eqnarray*}
\omega_6 &=& -2^{n-2}+\frac{\omega_0-1}{2}=-2^{n-3}+\frac{2^n - 1}
{6}+ \frac{K(1)}{8}-\frac{1}{2}\\
&=& \frac{2^{n-2}-4}{6}+\frac{K(1)}{8}.
\end{eqnarray*}
Finally, it can be proved that \(\omega_6 \geq 1\) for any \(n \geq 6\), implying that \(G_3\) is differentially \(6\)-uniform. Actually, it has been proved in~\cite[Th.~3.4]{Lachaud_Wolfmann90} that 
\[-2^{\frac{n}{2}+1}+1 \leq K(1) \leq 2^{\frac{n}{2}+1}+1\;\]
implying that \(\omega_6 >0\) when \(n > 5\).
It is worth noticing that \(G_3\) is APN when \(n=5\) since its inverse is 
the quadratic APN permutation \(x \mapsto x^9\). When \(n=4\), \(G_3\) is
locally-APN, and not APN,  since it corresponds to the inverse function over 
\({\mathbb F}_{2^4}\). 
\end{proof}

By combining the previous theorem and Corollary~\ref{cormain}, we deduce the differential spectrum of  \(G_{n-2}: x \mapsto x^{2^{n-2}-1}\) over \(\F\). 
\begin{corollary}\label{coro5}
Let \(G_{n-2}: x \mapsto x^{2^{n-2}-1}\) over \(\F\) with \(n \geq 6\). 
Then, we have:
\begin{itemize}
\item if \(\gcd(n,3)=1\), \(G_{n-2}\) is differentially \(6\)-uniform and for any \(b \in \F\), \(\delta(b) \in \{0, 2, 6\}\). Moreover, its differential spectrum is given by:
\begin{eqnarray*}
\omega_6 & = & \left\{\begin{array}{ll}\frac{2^{n-2}+1}{6} - \frac{K(1)}{8} & \mbox{ for odd \(n\)}  \\
\frac{2^{n-2}-4}{6} + \frac{K(1)}{8} & \mbox{ for even \(n\)}  \end{array}\right. \\
\omega_2 & = & 2^{n-1}-3 \omega_6 \\
\omega_0 & = & 2^{n-1} + 2 \omega_6\;;
\end{eqnarray*}

\item if \(3\) divides \(n\), \(G_{n-2}\) is differentially \(8\)-uniform and for any \(b \in \F\), \(\delta(b) \in \{0, 2, 6,8\}\). Moreover, its differential spectrum is given by:
\begin{eqnarray*}
\omega_8 & = & 1 \\
\omega_6 & = & \left\{\begin{array}{ll}\frac{2^{n-2}-5}{6} - \frac{K(1)}{8}  & \mbox{ for odd \(n\)}  \\
\frac{2^{n-2}-10}{6} + \frac{K(1)}{8} & \mbox{ for even \(n\)}  \end{array}\right. \\
\omega_2 & = & 2^{n-1}-3 \omega_6 -4\\
\omega_0 & = & 2^{n-1} + 2 \omega_6+3\;;
\end{eqnarray*}
\end{itemize}
\end{corollary}
\begin{proof}
Let \((\omega_0', \omega_2', \omega_4', \omega_6')\) denote the differential spectrum of \(G_3\) over \(\F\). We apply Corollary \ref{cormain} (with $s=3$).
Then, if \(\gcd(3,n)=1\), \(\delta_3(0)=0\) and \(\delta_{n-2}(1)=2\). Otherwise, \(\delta_3(0)=6\) and \(\delta_{n-2}(1)=8\). Moreover, in both cases, \(\delta_3(1)=4\) for \(n\) even and \(\delta_3(1)=2\) for \(n\) odd. 
It follows that, 
\begin{itemize}
\item for  \(\gcd(3,n)=1\), \(n\) odd, we have \((\delta_3(0), \delta_3(1)) = (0,2)\) and \\\((\delta_{n-2}(0), \delta_{n-2}(1)) = (0,2)\). Then, \(\omega_i = \omega'_i\) for all \(i\);
\item for  \(\gcd(3,n)=1\), \(n\) even, we have \((\delta_3(0), \delta_3(1)) = (0,4)\) and \\\((\delta_{n-2}(0), \delta_{n-2}(1)) = (2,2)\). Then, \(\omega_0 = \omega_0'-1\), \(\omega_4 = \omega_4'-1\) and \(\omega_2 = \omega_2'+2\).
\item for  \(\gcd(3,n)=3\), \(n\) odd, we have \((\delta_3(0), \delta_3(1)) = (6,2)\) and \\\((\delta_{n-2}(0), \delta_{n-2}(1)) = (0,8)\).  Then, \(\omega_8=1\), \(\omega_6 = \omega_6'-1\), \(\omega_2 = \omega_2'-1\) and \(\omega_0 = \omega_0'+ 1\).
\item for  \(\gcd(3,n)=3\), \(n\) even, we have \((\delta_3(0), \delta_3(1)) = (6,4)\) and \\\((\delta_{n-2}(0), \delta_{n-2}(1)) = (2,8)\).  Then,  \(\omega_8=1\), \(\omega_6 = \omega_6'-1\), \(\omega_4 = \omega_4'-1\), \(\omega_2 = \omega_2'+1\) and \(\omega_0 = \omega_0'\).
\end{itemize}
The result finally follows from Theorem~\ref{th:x^7}.
\end{proof}
 
The minimum distance of the cyclic code of length $2^n-1$ with defining set $\{1,7\}$
has been studied by van Lint and Wilson in~\cite{vlw2}. More precisely, they have proved that this code has minimum distance at most~$4$ for \(n \geq 6\). The previous corollary recovers this result and also provides the exact number of codewords of weight~\(3\) and~\(4\) in this code.
\begin{corollary}
Let \(B_3\) (resp. \(B_4\)) denote the number of codewords of Hamming weight~\(3\) (resp. of Hamming weight~\(4\)) in the binary cyclic code of length $2^n-1$ with defining set $\{1,7\}$. Then, we have
\begin{itemize}
\item if \(n\) is odd
\begin{eqnarray*}
B_3 & = & 0 \\
B_4 & = & (2^n-1) \left(\frac{2^{n-2}+1}{6} - \frac{K(1)}{8}\right);
\end{eqnarray*}
\item if \(n\) is even
\begin{eqnarray*}
B_3 & = & \frac{(2^n-1)}{3} \\
B_4 & =& (2^n-1) \left(\frac{2^{n-2}-4}{6} + \frac{K(1)}{8}\right).
\end{eqnarray*}
\end{itemize}
\end{corollary}
\begin{proof}
Let \(F(x)=x^d\) over \(\F\) and let \(\delta(b)\), \(b \in \F\), denote the number of solutions~\(x\) of 
\[D_1 F(x) = F(x+1) + F(x) = b\;.\]
It is known from Proposition~2 and Lemma~2 in~\cite{BCC} that the number of codewords of weight \(3\) and~\(4\) in the cyclic code of length~\((2^n-1)\) with defining set \(\{1,d\}\) is given by
\begin{eqnarray*}
B_3 & = & \frac{(2^n-1)}{6} (\delta(1)-2)\\
B_3 + B_4 & = & \frac{(2^n-1)}{24} \left[\#\{(x,y) \in \F \!\times \!\F: D_1 F(x) \!=\! D_1 F(y)\} - 2^{n+1}\right].
\end{eqnarray*}
Therefore, we have
\begin{eqnarray*}
B_3 + B_4 & = & \frac{(2^n-1)}{24}  \left(\sum_{b \in \F} \delta(b)^2 - 2^{n+1}\right)
\\
& = & \frac{(2^n-1)}{24}  \left[\sum_{i=0}^{2^n} i^2 \omega_i - 2^{n+1}\right]\;.
\end{eqnarray*}
This formula was proved in Corollary~1 of~\cite{BCC}, but only in the particular case where \(\gcd(d,2^n-1)=1\).
For \(d=7\), Theorem~\ref{th:x^7} implies that
\[B_3+B_4 = (2^n-1) \omega_6.\]
Then, the values of~\(B_3\) and~\(B_4\) are deduced from the expression of \(\omega_6\) given in Theorem~\ref{th:x^7}.
\end{proof}

\subsection{Exponents $\boldsymbol{2^{\lfloor n /2 \rfloor}-1}$}\label{sec-n/2}

We are going to determine the differential uniformity of \(G_t\) 
for \(t = \lfloor n /2 \rfloor\). We first consider the case where 
\(n\) is even. Note that in this case, \(G_t\) is not  a permutation
since $2^n-1=(2^t-1)(2^t+1)$.
\begin{theorem}\label{th:n/2}
Let  \(n\) be an even integer,  \(n >  4\) and  \(G_{t}(x)=x^{2^{t}-1}\)
 for \(t = \frac{n}{2}\). Then  \(G_{t}\) is locally-APN. More precisely 
\[
\delta(G_{t}) = 2^{t}-2 ~~\mbox{and}~~\delta(b)\leq 2,~\forall~
b\in\F\setminus\FB. 
\]
Moreover, the differential spectrum of \(G_{t}\) is:

\medskip
\noindent
$\bullet$ if \(n \equiv 0 \bmod{4}\) then
\begin{eqnarray*}
\omega_{2^{t}-2} & = & 1\\
\omega_{i}& = & 0 ,~\forall~i,~ 2 <  i <2^{t}-2\\
\omega_2 & = & 2^{n-1}-2^{t-1} + 1\\
\omega_0 & = & 2^{n-1}+2^{t-1} -2;
\end{eqnarray*}
$\bullet$ if \(n \equiv 2 \bmod{4}\),
\begin{eqnarray*}
\omega_{2^{t}-2} & = & 1\\
\omega_{i}& = & 0 ,~\forall~i,~ 4 <  i <2^{t}-2\\
\omega_4 & = & 1\\
\omega_2 & = & 2^{n-1}-2^{t-1} - 1\\
\omega_0 & = & 2^{n-1}+2^{t-1} -1.
\end{eqnarray*}
\end{theorem}
\begin{proof}
From Theorem~\ref{th: 2^t-1}, we obtain directly \(\delta(0) = 2^{t}-2\).
Also, \(\delta(1)=2\) if \(t\) is even and \(\delta(1)=4\) otherwise.

Now, for all \(b \not \in \FB\), we have to determine the number of roots in 
\(\F\) of \(P_b(x)= x^{2^{t}} + bx^2+(b+1)x\) or, equivalently, the
number of  roots of 
\begin{equation}\label{eqth7}
\left(P_b(x)\right)^{2^{t}}= x + b^{2^{t}}x^{2^{t+1}}+
(b+1)^{2^{t}}x^{2^{t}}.
\end{equation}
If $x$ is a root of $P_b$ then $x^{2^{t}}= bx^2+(b+1)x$. So, $P_b(x)=0$ implies 
\begin{eqnarray*}
\left(P_b(x)\right)^{2^{t}} & = & x +  b^{2^{t}}(x^{2^{t}})^2+
(b^{2^t}+1)x^{2^{t}}\\
& = & x +  b^{2^{t}}(bx^2+(b+1)x)^2+(b^{2^t}+1)(bx^2+(b+1)x)\nonumber\\
& = & b^{2^t+2} x^4 + (b^{2^{t}+2}+b^{2^{t}+1} + b^{2^{t}}+b)
 x^2 + (b^{2^{t}+1}+b^{2^t}+b) x\nonumber\\
& =& b^{2^t+2}(x^2+x)^2+ (b^{2^{t}+1} + b^{2^{t}}+b)(x^2+x).\nonumber
\end{eqnarray*}
Thus, we get a linear polynomial of degree~\(4\) which has at
least the roots $0$ and $1$. Hence, this polynomial has $\tau$ roots
where $\tau$ is either \(4\) or \(2\), including \(x=0\) and
\(x=1\). Therefore, for any $b\not\in\FB$, $\delta(b)\leq 2$ since
$\delta(b) \leq\tau-2$.  We deduce that \(G_{t}\) is localy-APN.

We also proved that  $\omega_i=0$ unless $i\in\{0,2, 2^{t}-2\}$
when $t$ is even and $i\in\{0,2,4, 2^{t}-2\}$ otherwise.
Moreover $\omega_{2^{t}-2}=\omega_4=1$.
According to Lemma \ref{ident}, we have for $t$ even :
\[
2^n=\omega_0+\omega_2+\omega_{2^{t}-2}=\omega_0+\omega_2+1
\]
and 
\[
2^n=2\omega_2+(2^{t}-2)\omega_{2^{t}-2}=2\omega_2+(2^{t}-2).
\]
So, we get $\omega_2=2^{n-1}-2^{t-1}+1$ and conclude with 
$\omega_0=2^n-\omega_2-1$. We proceed similarly for odd $t$,
with the following equalities derived from Lemma~\ref{ident}:
\[
2^n=\omega_0+\omega_2+2~~\mbox{and}~~2^n=2\omega_2+2^{t}+2.
\]
\end{proof}

And we directly deduce a property on the corresponding class of 
linear polynomials. 
\begin{corollary}
Let $n=2t$ and let $Tr_t$ denote the absolute trace on $\FF_{2^t}$.
Consider the polynomials over $\F$:
\[
x^{2^t}+bx^2+(b+1)x~~\mbox{and}~~x^{2^{t+1}}+bx^2+(b+1)x.
\]
Then, for any $b\in\F\setminus \FB$, these polynomials have
either $2$ or $4$ roots in~\(\F\).
The first one has $4$ roots if and only if $Tr_t(b^{-(2^t+1)})=1$ with 
$(1+b)\not\in {\cal G}$, where $\cal G$ is the cyclic subgroup of $\F$ of order $2^t+1$. 
\end{corollary}
\begin{proof}
Let \(P_b(x) = x^{2^t}+bx^2+(b+1)x\). We define
\[Q_b(x) = \frac{(P_b(x))^{2^t} + b^{2^t} (P_b(x))^2 + (b^{2^t}+1) P_b(x)}{b^{2^t+2}(x^2+x)}\]
Using~(\ref{eqth7}), we get :
\[
Q_b(x)=x^2+x+A, ~\mbox{with}~A=\frac{b^{2^{t}+1} + b^{2^{t}}+b} {b^{2^t+2}}.
\]
To be clear, we summarize the situation:\\
--  if $P_b(x)=0$, $x\not\in \{0,1\}$, then $Q_b(x)=0$;\\
--  when $Q_b(x)=0$, $x\not\in \{0,1\}$, one can have $P_b(x)\neq 0$;\\
-- if $Q_b(x)=0$ for $x\in \{0,1\}$ only, this holds for $P_b(x)$ too.

\medskip
We consider the case where $P_b(x)=0$ has more than the two solutions
$0$ and $1$. The equation  $Q_b(x)=0$ has two solutions (not in $\{0,1\}$) 
if and only if $Tr(A)=0$ with $A\neq 0$. But
\[
Tr(A)=Tr\left(\frac{1}{b}+\frac{1}{b^2}+\frac{1} {b^{2^t+1}}\right)
=Tr\left(\frac{1} {b^{2^t+1}}\right)=0,
\]
for all $b$, since $b^{2^t+1}\in \FF_{2^t}$. And,
\[
A\neq 0~\Leftrightarrow~b^{2^{t}+1} + b^{2^{t}}+b\neq 0
\Leftrightarrow~(b+1)^{2^t+1}\neq 1,
\]
that is : $b+1$ is not in the cyclic subgroup $\cal G$
 of order $2^t+1$ of $\F^*$.
On the other hand, if $Q_b(x)=0$ then $x^2+x=A$ and we get
\begin{eqnarray*}
P_b(x) &=&  x^{2^{t}}+x + b(x^2+x)\\
&=&  (x^2+x)^{2^{t-1}}+(x^2+x)^{2^{t-2}}+\dots+(x^2+x)+ b(x^2+x)\\
&=&  A^{2^{t-1}}+\dots+A+bA.
\end{eqnarray*}
We compute this last expression by replacing the value of $A$:
\begin{eqnarray*}
P_b(x) &=&  \sum_{i=0}^{t-1}\left(\frac{1}{b}+\frac{1}{b^2}\right)^{2^i}+
\sum_{i=0}^{t-1}\left(\frac{1} {b^{2^t+1}}\right)^{2^i}+1+
\frac{1}{b}+\frac{1} {b^{2^t}}\\
 &=& Tr_t\left(\frac{1} {b^{2^t+1}}\right)+1,
\end{eqnarray*}
where $Tr_t$ is the absolute trace on $\FF_{2^t}$. We conclude that
 $P_b(x)=0$ if and only if $Tr_t( {b^{-(2^t+1)}})=1$, with $b+1\not\in{\cal G}$.
\end{proof}
According to  Corollary~\ref{cormain}, the differential spectrum of 
\(x \mapsto x^{2^{\frac{n}{2}}-1}\) determines the differential spectrum of 
\(x \mapsto x^{2^{\frac{n}{2}+1}-1}\)

\begin{theorem}
Let  \(n\) be an even integer  \(n >  4\) and  \(G_{t+1}(x)=x^{2^{t+1}-1}\)
 for \(t = \frac{n}{2}\). 
Then, \(G_{t+1}\) is locally-APN. It is differentially \(2^t\)-uniform and
 its differential spectrum is
\begin{eqnarray*}
\omega_{2^{t}} & = & 1\\
\omega_{i}& = & 0 ,~\forall~i,~ 2 <  i <2^{t}\\
\omega_2 & = & 2^{n-1}-2^{t-1} \\
\omega_0 & = & 2^{n-1}+2^{t-1} -1\;.
\end{eqnarray*}
Moreover, \(G_{t+1}\) is a permutation if and only if \(n \equiv 0 \bmod{4}\).
\end{theorem}
\begin{proof}
First, since \(n=2t\), we have \(\gcd(t+1,n)=1\) if \(t\) is even ({\em i.e.},
 \(n \equiv 0 \bmod{4}\)) and \(\gcd(t+1,n)=2\) if \(t\) is odd ({\em i.e.},
 \(n \equiv 2 \bmod{4}\)). Here $s=t+1$.

Let \((\omega_i')_{0 \leq i \leq 2^n}\) (resp. \((\omega_i)_{0 \leq i \leq 2^n}\))
 denote the differential spectrum of \(G_t\) (resp. \(G_{t+1}\)) over \(\F\).
\begin{itemize}
\item For \(n \equiv 0 \bmod{4}\), we have
\((\delta_t(0), \delta_t(1)) = (2^t-2,2)\) and \((\delta_s(0), \delta_s(1)) = (0,2^t)\). Thus, \(\omega_0 = \omega_0'+1\), \(\omega_2 = \omega_2'-1\), \(\omega_{2^t-2} = \omega_{2^t-2}'-1\) and \(\omega_{2^t} = 1\).
\item For \(n \equiv 2 \bmod{4}\), we have \((\delta_t(0), \delta_t(1)) = (2^t-2,4)\) and \((\delta_s(0), \delta_s(1)) = (2,2^t)\). Thus, \(\omega_2 = \omega_2'+1\), \(\omega_4 = \omega_4'-1\), \(\omega_{2^t-2} = \omega_{2^t-2}'-1\) and \(\omega_{2^t} = 1\).
\end{itemize}
The differential spectrum of \(G_{t+1}\) is then directly deduced by combining the previous formulas with the values of \(\omega_i'\) computed in Theorem~\ref{th:n/2}.
\end{proof}

\medskip
In the case where \(n\) is odd, the differential uniformity of 
\(G_t\), with $t=\frac{n-1}{2}$,  can also be determined.

\begin{theorem}
Let \(n\) be an odd integer, \(n > 3\). Let $G_t(x)=x^{2^{t}-1}$ with \(t = (n-1)/2\).
Then, \(G_t\) is a permutation and for all $b\in\F\setminus \FB$
 we have $\delta(b)\in\{0,2,6\}$. Moreover
\begin{itemize}
\item if \(n \equiv 0 \bmod{3}\), then \(\delta(G_t)=8\), and the 
differential spectrum satisfies \(\omega_i=0\) for all \(i \not \in 
\{ 0,2, 6, 8\}\) and \(\omega_8=1\).
\item if \(n \not \equiv 0 \bmod{3}\), then \(\delta(G_t)\leq 6\) and the
 differential spectrum satisfies \(\omega_i=0\) for all \(i \not \in
 \{ 0,2, 6\}\).
\end{itemize}
\end{theorem}
\begin{proof}
From Theorem~\ref{th: 2^t-1},  we have \(\delta(0) = 0\); moreover,
  if \(3\) divides \(n\) then \(\delta(1)=8\) else \(\delta(1)=2\).
Now, for all \(b \not \in \FB\), we have to determine the number of roots in 
\(\F\) of 
\[
P_b(x)= x^{2^{t}} + bx^2+(b+1)x,
\] 
or, equivalently, the number of roots of 
\[
\left(P_b(x)\right)^{2^{t+1}}= x + b^{2^{t+1}}x^{2^{t+2}}+
(b+1)^{2^{t+1}}x^{2^{t+1}}.
\]
Set $c=b^{2^{t+1}}$ and $Q_b(x)=\left(P_b(x)\right)^{2^{t+1}}$.
If $x$ is a root of $P_b$ then $x^{2^{t}}= bx^2+(b+1)x$. So, $P_b(x)=0$ implies
\begin{eqnarray*}
Q_b(x) & = &  x + c(x^{2^{t}})^4+(c+1)(x^{2^{t}})^2\\
& = & x + c(bx^2+(b+1)x)^4+(c+1)(bx^2+(b+1)x)^2\\
& = & cb^4x^8+(c(b+1)^4+(c+1)b^2)x^4+(c+1)(b^2+1) x^2 +x\;.
\end{eqnarray*}
Since \(Q_b\) has degree~\(8\),  it has either \(8\) or \(4\) or \(2\)
solutions. In other terms, $\delta(b)\in\{0,2,6\}$.
\end{proof}

\section{Conclusions}
In this work, we  point out that the family of all power functions
\begin{equation}\label{famil}
\{~G_t~:~x \mapsto x^{2^t-1} \mbox{ over } \F, \; 1 < t <n\}
\end{equation}
has interesting differential properties. 
The study of these properties led us to introduce locally-APN
functions, as a generalization of the differential spectrum 
of the inverse function.

In particular, we give several results about
the functions  with a low differential uniformity within family (\ref{famil}).
There are classes of functions $G_t$ such that $\delta(G_t)=6$.
It is the case for the functions $G_3$ over $\F$ (see Theorem \ref{th:x^7}).

The functions such that $\delta(G_t)\leq 4$ can be differentially
$4$-uniform for even $n$ only (see Corollary \ref{coro3}).
We have shown that, for exponents of the form \(2^t-1\), the APN property
 imposes many conditions of the value of \(t\).In particular, it is easy to prove, using Theorem \ref{th: 2^t-1} that such exponent must satisfy $\gcd(t,n)=2$
 for even \(n\) and $\gcd(t,n)=\gcd(t-1,n)=1$ for odd~\(n\). 
Another condition can be derived from the recent result by Aubry and 
Rodier~\cite{AR} who proved the following theorem.
\begin{theorem}{\rm \cite[Theorem 9]{AR}}\label{aubrod}
Let $G_t: x \mapsto x^{2^t-1}$ over $\F$ with $t\geq 3$. 
If $7\leq 2^t-1 < 2^{n/4}+4.6$ then $\delta(G_t)>4$.
\end{theorem}
Thanks to Corollary \ref{cormain},  we can extend this result as follows.
\begin{corollary}
Let $G_t: x \mapsto x^{2^t-1}$ over $\F$ with $ 3 \leq t\leq n-2$. If $\delta(G_t)\leq 4$, then
\[\log_2(2^{\frac{n}{4}} + 5.6) \leq t \leq n+1 - \log_2(2^{\frac{n}{4}} + 5.6)\;.\]
\end{corollary}

\begin{proof}
Let $s=n-t+1$ so that $3\leq s\leq n-2$.
In this proof, we denote by $\delta_t(b)$ (resp. $\delta_s(b)$)
the quantities $\delta(b)$ corresponding to $G_t$ (resp. $G_s$).

From Theorem \ref{aubrod}, we know that $\delta(G_t)\leq 4$
implies 
\[
2^{n/4}+4.6\leq 2^t-1, ~\mbox{{\em i.e,} $t\geq \log_2(2^{\frac{n}{4}} + 5.6)$.}
\]
We consider now the function  $G_s$. Note that, from 
Theorem \ref{th: 2^t-1}, $\delta(G_t)\leq 4$ implies 
$\delta_t(0)\in\{0,2\}$ and  $\delta_t(1)\in\{2,4\}$.
Moreover, we obtain directly from Corollary \ref{cormain} :
\begin{itemize}
\item   $\delta_s(b)\leq 4$,  for any $b\not\in\FB$.
\item $\delta_s(0)\in\{0,2\}$ and $\delta_s(1)\in\{2,4\}$.
\end{itemize}
Thus $\delta(G_s)\leq 4$ and, applying Theorem \ref{aubrod} again, we get
\[
s\geq \log_2(2^{\frac{n}{4}} + 5.6),~ ~\mbox{{\em i.e,}
$n+1-\log_2(2^{\frac{n}{4}} + 5.6)\geq t$.}
\]
\end{proof}
We now concentrate  on APN functions belonging to the family (\ref{famil}).
Some are well-known as the inverse
 permutation for $n$ odd ($t=n-1$) and the quadratic function $x \mapsto x^3$
 ($t=2$). There is also the function $G_t$ for $t=(n+1)/2$ with $n$ odd,
because this function is the inverse of  the quadratic function
$x \mapsto x^{2^{(n+1)/{2}}+1}$. Recall that $x^{2^i+1}$ is an APN 
function over $\F$ if and only if  $\gcd(n,i)=1$ and we have obviously
$\gcd(n,(n+1)/2)=1$ (for odd $n$).
We conjecture that these three functions are the only APN functions
 within  family (\ref{famil}). 
\begin{conjecture}
Let $G_t(x)=x^{2^t-1}$, $2\leq t\leq n-1$. If 
$G_t$ is  APN then either $t=2$ or  $n$ is odd and 
$t\in \{\frac{n+1}{2},n-1\}$.
\end{conjecture}
If the previous conjecture holds then there are some consequences
for the functions of (\ref{famil}) which are differentially \(4\)-uniform.
From Corollary~\ref{coro3}, we can say that such a function $G_t$ 
is a function over $\F$ with $n$ even. Moreover $G_s$, $s=n-t+1$, is APN.
If the conjecture holds then $s=2$ ($t=n-1$) is the only one possibility.
So, in this case we could conclude that {\em the inverse function
is the only one differentially \(4\)-uniform function 
of family (\ref{famil})}.

\end{document}